\documentclass[fleqn,12pt,twoside]{article}

\usepackage[headings]{espcrc1}
\readRCS $Id: espcrc1.tex,v 1.2 2004/02/24 11:22:11 spepping Exp $
\usepackage[dvips]{graphicx}
\usepackage[figuresright]{rotating}

\newtheorem{theorem}{Theorem}

\newtheorem{conjecture}{Conjecture}
\newtheorem{corollary}{Corollary}

\newtheorem{lemma}{Lemma}

\newenvironment{proof}[1][Proof.]{\begin{trivlist}
\item[\hskip \labelsep {\bfseries #1}]}{\end{trivlist}}

\newenvironment{acknowledgement}[1][Acknowledgement]{\begin{trivlist}
\item[\hskip \labelsep {\bfseries #1}]}{\end{trivlist}}

\newcommand{\AmS}{{\protect\the\textfont2
  A\kern-.1667em\lower.5ex\hbox{M}\kern-.125emS}}

\usepackage{amsmath}
\usepackage{amsfonts}
\title{Maximum $\Delta$-edge-colorable subgraphs of class II graphs}

\author{Vahan V. Mkrtchyan{\thanks{email: vahanmkrtchyan2002@\{ysu.am, ipia.sci.am,
yahoo.com\}}{\thanks{The author is supported by a fellowship from Heinrich Hertz-Stiftung}}, 
Eckhard Steffen{\thanks{email: es@upb.de}} \\ 
Paderborn Institute for Advanced Studies in Computer Science and Engineering,
Paderborn University, Warburger Str. 100, 33098 Paderborn, Germany}}

\date{}

\runtitle{Maximum $\Delta$-edge-colorable subgraphs}
\runauthor{Vahan V. Mkrtchyan, Eckhard Steffen}

\begin{document}

\maketitle

\begin{abstract}
{\bf Abstract:} A graph $G$ is class II, if its chromatic index is at least $\Delta+1$. Let $H$ be a maximum $\Delta$-edge-colorable subgraph of $G$.
The paper proves best possible lower bounds for $\frac{|E(H)|}{|E(G)|}$, and structural properties of maximum $\Delta$-edge-colorable subgraphs. 
It is shown that every set of vertex-disjoint cycles of a class II graph with $\Delta\geq3$ can be
extended to a maximum $\Delta$-edge-colorable subgraph. Simple graphs have a maximum
$\Delta$-edge-colorable subgraph such that the complement is a matching. Furthermore, a maximum $\Delta$-edge-colorable subgraph of 
a simple graph is always class I.\\[.2cm]
\end{abstract}

\section{Introduction}

We consider finite, undirected graphs $G=(V,E)$ with vertex set $V=V(G)$
and edge set $E=E(G)$. The graphs might have multiple edges but no loops.

Let $G$ be a graph. The length of the shortest (odd) cycle of the
underlying simple graph of $G$ is called the {\em (odd) girth} of $G$, i.e.~the girth of $G$ is always at least three.
For $X\subseteq V(G)$ we denote by $\partial_{G}(X)$
the set of edges with precisely one end in $X$. The minimum and maximum degree
of $G$ is denoted by $\delta(G)$ and $\Delta(G)$, respectively. A partial proper $t$-edge-coloring of $G$ is an assignment
of colors $1,...,t$ to some edges of $G$ such that adjacent
edges receive different colors. Let $\theta$ be a partial proper $t$-edge-coloring
of $G$. The components of the subgraph which is induced by two colors $\alpha$ and $\beta$ are 
called $(\alpha,\beta)$-Kempe-chains. That is, a Kempe-chain is either a path or an even cycle. In the case that it is a path $P$, we 
also say,  that $P$ is an $\alpha$-$\beta$-alternating path.

A partial proper $t$-edge-coloring of $G$ is called a proper
$t$-edge-coloring (or just $t$-edge-coloring) if all edges
are assigned some color. The smallest number $k$ for which $G$ has
a $k$-edge-coloring is called the {\em chromatic index} of $G$, and it is
denoted by $\chi'(G)$. A graph $G$ is {\em critical}, if 
$\chi'(G)>\Delta(G)$ and $\chi'(G-e)<\chi'(G)$, for every edge $e \in E(G)$.
If $G$ is simple, we also say that $G$ is $\Delta(G)${\em-critical}.
Clearly, $\Delta(G) \leq \chi'(G)$, and by the classical theorems of Shannon and Vizing we have the following upper bounds
for the chromatic index of a graph. 

\begin{theorem}\label{Shannon}(Shannon) Let $G$ be a graph, then $\chi'(G)\leq\lfloor\frac{3\Delta(G)}{2}\rfloor$.
\end{theorem}

\begin{theorem}\label{Vizing}(Vizing) Let $G$ be a graph, then 
$\chi'(G)\leq\Delta(G)+\mu(G)$, where $\mu(G)$ is the maximum 
multiplicity of an edge in $G$. 
\end{theorem}

A graph $G$ with $\chi'(G) = \Delta(G) = \Delta$ is {\em class I}, 
otherwise it is {\em class II}. 
There are long standing open conjectures on class II graphs, cf. \cite{Jensen_Toft}. 
It is a notorious difficult open problem to characterize class II graphs or even to obtain 
some insight into their structural properties. This paper focuses on the $\Delta$-edge-colorable part
of graphs. A subgraph $H$ of $G$ is called {\em maximum} $\Delta$-edge-colorable, if it is $\Delta$-edge-colorable
and contains as many edges as possible. The fraction $|E(H)|/|E(G)|$ is the subject of many papers,
e.g. lower bounds are proved for cubic, subcubic or 4-regular graphs, \cite{A_Haas,part1,Rizzi_2009}. 
One aim of this paper is to prove a general best possible lower bound for all graphs.

Let $H$ be a maximum $\Delta$-edge-colorable
subgraph of $G$, which is properly colored with colors $1,...,\Delta$.
Usually, we will refer to edges of $E(G)\backslash E(H)$ as uncolored
edges. For a vertex $v$ of $G$ let $C(v)$ be the set of colors
that appear at $v$, and $\overline{C}(v)=\{1,...,\Delta\}\backslash C(v)$
be the set of colors which are missing at $v$.

Let $e=(v,u)\in E(G)\backslash E(H)$ be an uncolored edge, $\alpha\in\overline{C}(u)$,
$\beta\in\overline{C}(v)$. Since $H$ is a maximum $\Delta$-edge-colorable
subgraph of $G$, we have that $\alpha\in C(v)$ and $\beta\in C(u)$.
Consider the $\alpha$-$\beta$-alternating path $P$ starting from the
vertex $v$. Again, since $H$ is a maximum $\Delta$-edge-colorable
subgraph of $G$, the path $P$ ends in $u$. Thus $P$ is an even
path, which together with the edge $e$ forms an odd cycle. We will
denote this cycle by $C_{\alpha,\beta,H}^{e}$. If the subgraph $H$
is fixed, then we will shorten the notation to $C_{\alpha,\beta}^{e}$.

Kempe chains forming an odd cycle together with an uncolored edge $e$, $C_{\alpha,\beta}^e$, play a central role in the study of cubic graphs \cite{SteffenClassification,Steffen}. The second aim of this paper is to generalize
some of these results to arbitrary graphs, and to investigate the
maximum $\Delta$-edge-colorable subgraphs. We show that any set of
vertex-disjoint cycles of a graph $G$ with $\Delta(G)\geq3$ can be extended to a maximum
$\Delta$-edge-colorable subgraph of $G$. In particular, any $2$-factor
of a graph with maximum degree at least three can be extended to such a subgraph.

Let $\phi$ be a $\chi'(G)$-edge-coloring of a graph $G$ with
$\chi'(G)=\Delta(G)+k$ ($k\geq1$), and $r'_{\phi}(G)=\min\sum_{j=1}^{k}|\phi^{-1}(i_{j})|$.
We define $r'_{e}(G)=\min_{\phi}r'_{\phi}(G)$ as the minimum size
of the union of $k$ color-classes in a $\chi'(G)$-edge-coloring
of $G$. Let $r_{e}(G)$ denote the minimum
number of edges that should be removed from $G$ in order
to obtain a graph $H$ with $\chi'(H)=\Delta(G)$.
Clearly, $r_{e}(G)=|E(G)|-|E(H)|$, where $H$ is a maximum $\Delta$-edge-colorable
subgraph of $G$. In  \cite{part2} it is shown that the complement of any maximum
$3$-edge-colorable subgraph of a cubic graph is a matching, and hence $r_{e}(G)=r'_{e}(G)$
for cubic graphs.  This paper generalizes this result to simple graphs. 
We further prove some bounds for the vertex degrees of a maximum $\Delta(G)$-edge-colorable
subgraph $H$.

\section{Maximum $\Delta$-edge-colorable subgraphs: Cycles}

The key property of cycles corresponding to uncolored edges in cubic graphs that is
used in \cite{SteffenClassification,Steffen} is their vertex-disjointness.
As the example from Figure \ref{FatTriangleExample} shows, they can
have even common edges in the general case. To see this, consider
the graph $G$ and its maximum $\Delta(G)$-edge-colorable subgraph
$H$ from Figure \ref{FatTriangleExample}. Note that $\overline{C}(a)=\{4\}$,
$\overline{C}(b)=\{3\}$, $\overline{C}(c)=\{1,2\}$, 
and hence $E(C_{1,3}^{(b,c)})\cap E(C_{1,4}^{(a,c)})\neq\emptyset$.

\begin{center}
\begin{figure}[ht]
\begin{center}
\includegraphics[height=10pc, width=13pc]{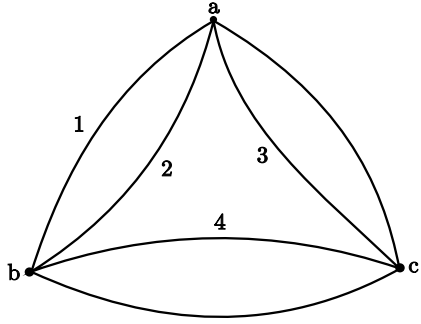}\\
\caption{Cycles of uncolored edges may intersect.}\label{FatTriangleExample}
\end{center}
\end{figure}
\end{center}

Despite this, it turns out that, as Theorem \ref{assignment} demonstrates
below, the edge-disjointness of the cycles can be preserved. For the
proof we need the following Lemma, where we implicitly assume that the maximum $\Delta(G)$-edge-colorable
subgraph $H$ of a graph $G$ is colored with $\Delta=\Delta(G)$
colors. Before we proceed, let us note that if $C_{\alpha,\beta}^{e}$ and $C_{\gamma,\delta}^{e'}$ are two cycles corresponding to two different uncolored edges $e$ and $e'$ with respect to a maximum $\Delta(G)$-edge-colorable subgraph, then $E(C_{\alpha,\beta}^{e})\cap E(C_{\gamma,\delta}^{e'})\neq \emptyset$ implies $\{\alpha,\beta\}\cap \{\gamma,\delta\} \neq \emptyset$. The Lemma, that we are going to prove, describes the placement of the edges of $e$ and $e'$.

\begin{lemma}\label{CyclesIntersection} Let $H$ be any maximum
$\Delta(G)$-edge-colorable subgraph of a graph $G$, and $e$, $e'$
be two uncolored edges. If $E(C_{\alpha,\beta}^{e})\cap E(C_{\alpha,\gamma}^{e'})\neq\emptyset$,
then there is a vertex $v$ such that $e$ and $e'$ are
incident to $v$ and $\alpha \in \overline{C}(v)$. \end{lemma}

\begin{proof} Let $e$ and $e'$ be two uncolored edges with respect to $H$, and
$E(C_{\alpha,\beta}^{e})\cap E(C_{\alpha,\gamma}^{e'})\neq \emptyset$. If $\beta=\gamma$, then $e$ and $e'$ are parallel, therefore the statement of the Lemma is trivial. Thus, we can assume that $\beta\neq\gamma$, and therefore every edge $e''\in E(C_{\alpha,\beta}^{e})\cap E(C_{\alpha,\gamma}^{e'})$ is colored with color $\alpha$. 

We first show that $e$ and $e'$ are adjacent. Assume to the contrary that
this is not the case. Let $P'$ be a subpath of $C_{\alpha,\gamma}^{e'}$ that connects $e'$ to $(v,w)=e''\in E(C_{\alpha,\beta}^{e})\cap E(C_{\alpha,\gamma}^{e'})$, where $e''$ is chosen so that $E(P')\cap E(C_{\alpha,\beta}^{e})\cap E(C_{\alpha,\gamma}^{e'})=\emptyset$. 
Note, that this is always possible. Suppose that $P'$ connects $e'$ to the vertex $w$. Let $P$ be a subpath of $C_{\alpha,\beta}^{e}$ that connects $e$ to the vertex $v$, and does not pass through $w$.

There are edges $f\in E(P+e)$,
$f'\in E(P'+e')$ which are adjacent to $e''$. If $f\neq e$, then $f$ is colored with
color $\beta$, and similarly, if $f'\neq e'$, $f'$ is colored with color $\gamma$. Moreover, $f$ and $f'$ do not share a vertex, $f$ is
incident to $v$, and $f'$ to $w$.

Now recolor $E(P)+e$ by leaving
$e''$ uncolored and color the remaining edges with colors $\alpha$,
$\beta$ alternately, to obtain another maximum $\Delta$-edge-colorable
subgraph $H'$of $G$.
By the choice of $P$ and $P'$, no edge of the subpath $P'$ of $C_{\alpha,\gamma}^{e'}$
is involved in the recoloring process. Thus $P'+e'$ can be colored
with colors $\alpha$, $\gamma$ and we obtain a $\Delta$-edge-colorable
subgraph $H^{*}$ with $|E(H^{*})|>|E(H)|$, contradicting our choice
of $H$. Thus $e$ and $e'$ are adjacent. 

If $e$ and $e'$ are parallel, then our statement follows easily. It remains to consider the case when $e$ and $e'$ share precisely
one vertex, say $v$. Assume to the contrary that $\alpha\not\in\overline{C}(v)$,
i.e. $\alpha\in C(v)$ . Then, interchanging colors $\alpha$ and $\beta$ in $C_{\alpha,\beta}^{e}$
allows us to color $e'$ with color $\alpha$, contradicting the maximality
of $H$. Hence it follows that $\alpha\in\overline{C}(v)$, and the statement
is proved.$\square$ \end{proof}

\begin{theorem}\label{assignment}Let $H$ be any maximum $\Delta(G)$-edge-colorable
subgraph of a graph $G$, and let $E(G)-E(H)=\{e_{i}=(u_{i},v_{i})|1\leq i\leq n\}$
be the set of uncolored edges. Assume that $H$ is properly edge-colored
with colors $1,\ldots,\Delta(G)$. Then there is an assignment of
colors $\alpha_{1}\in\overline{C}(u_{1}),\beta_{1}\in\overline{C}(v_{1}),\ldots,\alpha_{n}\in\overline{C}(u_{n}),\beta_{n}\in\overline{C}(v_{n})$
to the uncolored edges, such that $
E(C_{\alpha_{i},\beta_{i}}^{e_{i}})\cap E(C_{\alpha_{j},\beta_{j}}^{e_{j}})=\emptyset$,
for all $1\leq i<j\leq n$. \end{theorem}

\begin{proof} We prove the statement by induction on the number $n$
of uncolored edges.

If $n=1$, then the statement is trivial. 

Let $n \geq 2$, and $G$ be a graph with $|E(G)-E(H)|=n$. We
need to consider two cases.

Case 1: There is a $k$ ($1\leq k\leq n$), such that $\Delta(G-e_{k})=\Delta(G)$.
Then, $H$ is a maximum $\Delta(G-e_{k})$-edge-colorable
subgraph of a graph $G-e_{k}$. Thus by the induction hypothesis, there
are $\alpha_{1}\in\overline{C}(u_{1}),\beta_{1}\in\overline{C}(v_{1}),\ldots,\alpha_{k-1}\in\overline{C}(u_{k-1}),\beta_{k-1}\in\overline{C}(v_{k-1}),\alpha_{k+1}\in\overline{C}(u_{k+1}),\beta_{k+1}\in\overline{C}(v_{k+1}),\ldots,\alpha_{n}\in\overline{C}(u_{n}),\beta_{n}\in\overline{C}(v_{n})$
such that \[
E(C_{\alpha_{i},\beta_{i}}^{e_{i}})\cap E(C_{\alpha_{j},\beta_{j}}^{e_{j}})=\emptyset,
\mbox{ for all } 1\leq i<j\leq n \hspace{.5cm} (i,j\neq k). \] 
Define
$I_{k}=\{e_{i}|i\neq k\textrm{, and }e_{i} \textrm{ is incident to }u_{k}\}$, and $J_{k}=\{e_{i}|i\neq k\textrm{, and }e_{i}\textrm{ is incident to }v_{k}\}$.
Then 
$|\overline{C}(u_{k})|\geq 1+|I_{k}|$, $|\overline{C}(v_{k})|\geq 1+|J_{k}|$,
and  hence there are $\alpha_{k}\in\overline{C}(u_{k}),\beta_{k}\in\overline{C}(v_{k})$,
such that 
$\alpha_{k}\neq\alpha_{i},e_{i}\in I_{k}$ and $\beta_{k}\neq\beta_{j},e_{j}\in J_{k}$.
Choose the colors $\alpha_{k},\beta_{k}$ to create $C_{\alpha_{k},\beta_{k}}^{e_{k}}$, then Lemma
\ref{CyclesIntersection} implies that
$ E(C_{\alpha_{i},\beta_{i}}^{e_{i}})\cap E(C_{\alpha_{j},\beta_{j}}^{e_{j}})=\emptyset$, 
for all $1\leq i<j\leq n$.

Case 2: For each $k$ ($1\leq k\leq n$) we have $\Delta(G-e_{k})<\Delta(G)$.
Then, there is a vertex $v$ of maximum degree $\Delta(G)$, such that
$v$ is incident to all uncolored edges $e_{1},...,e_{n}$ (Figure
\ref{AllUncoloreds}).
Thus, without loss of generality, we can assume that $v=v_{1}=...=v_{n}$.
Let $u^{(1)},...,u^{(l)}$ ($l \leq n$) be all the neighbors of $v$ such that
$v$ and $u^{(i)}$ are connected by $k_{i}$ uncolored edges, $k_{i}\geq1$ , $i=1,...,l$
(Figure \ref{AllUncoloreds}). Clearly, 
$k_{1}+...+k_{l}=n$.
We may assume that the edges $e_{1},...,e_{k_{1}}$ are incident
to $u^{(1)}$, the edges $e_{k_{1}+1},...,e_{k_{1}+k_{2}}$ are incident to
$u^{(2)}$,..., the edges $e_{k_{1}+...+k_{l-1}+1},...,e_{n}$ are incident to
$u^{(l)}$. Moreover, let $\overline{C}(v)=\{c_{1},...,c_{n}\}$ and
$\overline{C}(u^{(j)})=\{c_{1}^{(j)},...,c_{t_{j}}^{(j)}\},1\leq j\leq l$.

\begin{center}
\begin{figure}[ht]
\begin{center}
\includegraphics[height=13pc, width=20pc]{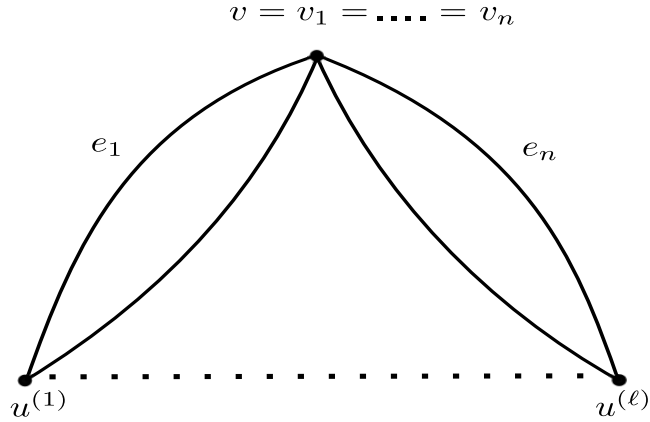}\\
\caption{All uncolored edges are incident to the vertex $v$.}
\label{AllUncoloreds}
\end{center}
\end{figure}
\end{center}

Note that $t_{j}\geq k_{j}$, for $j=1,...,l$. Define
\begin{eqnarray*}
\beta_{1}=c_{1},...,\beta_{n}=c_{n},\\
\alpha_{1}=c_{1}^{(1)},...,\alpha_{k_{1}}=c_{k_{1}}^{(1)},\\
\alpha_{k_{1}+1}=c_{1}^{(2)},...,\alpha_{k_{1}+k_{2}}=c_{k_{2}}^{(2)},\\
\dots\\
\alpha_{k_{1}+...+k_{l-1}+1}=c_{1}^{(l)},...,\alpha_{n}=c_{k_{l}}^{(l)}.\end{eqnarray*}
 Now, it is not hard to see that the choice of $\alpha_{1},\beta_{1},...,\alpha_{n},\beta_{n}$
and Lemma \ref{CyclesIntersection} imply that 
$E(C_{\alpha_{i},\beta_{i}}^{e_{i}})\cap E(C_{\alpha_{j},\beta_{j}}^{e_{j}})=\emptyset$,
for all $1\leq i<j\leq n$. $\square$
\end{proof}

The next Theorem generalizes the result of \cite{part2} that any
$2$-factor of a cubic graph can be extended to a maximum $3$-edge-colorable
subgraph to arbitrary graphs.

\begin{theorem}\label{CycleSystemExtension}Let $\overline{F}$ be any
set of vertex-disjoint cycles of a graph $G$ with $\Delta=\Delta(G)\geq3$. Then there is a maximum
$\Delta$-edge-colorable subgraph $H$ of $G$, such that $E(\overline{F})\subseteq E(H)$.
\end{theorem}

\begin{proof}Let $\Delta=\Delta(G)\geq3$. For $\overline{F}$ consider a maximum
$\Delta$-edge-colorable subgraph $H$ of $G$ with 
$|E(\overline{F})\cap E(H)|\textrm{ is maximum}$.

We will show that $E(\overline{F})\subseteq E(H)$. Assume to the contrary that there is an edge 
$e=(u,v)\in E(\overline{F})$, such that $e$ does not belong to $E(H)$.
Let us assume that $H$ is properly colored with colors $1,2,...,\Delta$.

Case 1: There are $\alpha\in C(u)\backslash C(v)$ and $\beta\in C(v)\backslash C(u)$,
such that $C_{\alpha,\beta}^{e}$ contains an edge $f$ that does
not belong to a cycle of $\overline{F}$.
Consider a proper partial $\Delta$-edge-coloring of $G$ obtained
from the coloring of $H$ by shifting the colors on the cycle $C_{\alpha,\beta}^{e}$
and leaving the edge $f$ uncolored. The new partial $\Delta$-edge-coloring
corresponds to a maximum $\Delta$-edge-colorable subgraph $H'$ of
$G$ with 
$|E(\overline{F})\cap E(H')|>|E(\overline{F})\cap E(H)|$,
contradicting the choice of $H$.

Case 2: For all $\alpha\in C(u)\backslash C(v)$ and $\beta\in C(v)\backslash C(u)$,
the edges of $C_{\alpha,\beta}^{e}$ lie on a cycle of $\overline{F}$.
Since the cycles of $\overline{F}$ are vertex-disjoint, it follows that
$|C(u)\backslash C(v)|=|C(v)\backslash C(u)|=1$ and there is only
one cycle $C_{\alpha,\beta}^{e}=C_{e}$, which, in its turn, is a cycle
of $\overline{F}$. This implies that $d_H(u)=d_H(v)=\Delta(G)-1\geq2$ and $d(u)=d(v)=\Delta(G)\geq3$.

Let us assume that $C_{e}-e$ is colored by the colors $\Delta-1$
and $\Delta$, alternately. For an edge $f$ of $C_e$, define $H_f$ obtained from $H$ as follows,

\begin{enumerate}
\item leave the edge $f$ uncolored, 
\item color the edges of the even path $C_{e}-f$ with colors $\Delta-1$
and $\Delta$, alternately, 
\item leave the colors of the rest of edges unchanged. 
\end{enumerate}

Note that the subgraphs $H_e$ and $H$ are the same, though they may have received different colorings, and for each edge $f$ of $C_e$, $H_f$ is a maximum $\Delta$-edge-colorable subgraph of $G$ with 
$|E(\overline{F})\cap E(H_f)|=|E(\overline{F})\cap E(H)|$. Similar to the consideration of the case 1 with respect to $H_f$, it can be shown that we can assume that if $f=(p,q)$ is an edge of $C_e$, then $d_{H_f}(p)=d_{H_f}(q)=\Delta(G)-1\geq2$ and $d(p)=d(q)=\Delta(G)\geq3$. This implies that the vertices of the cycle $C_e$ are of degree $\Delta$. Since $\Delta\geq3$, each vertex of $C_e$ is incident to $\Delta-2\geq 1$ edges that do not lie on $C_{e}$, and which are colored by the colors $1,...,\Delta-2$ in every $H_f$.
Let $\theta$ be a proper partial $\Delta$-edge-coloring
of $G$ which is obtained from the coloring of $H$ by deleting the
colors of all edges of $C_{e}$. Since $C_{e}$ is an odd cycle, it
follows that there is a $1$-$\Delta$-alternating path $P_{w}$ (with respect to the edge-coloring $\theta$) 
that starts from a vertex $w\in V(C_{e})$
and does not end at a vertex of $C_{e}$. Choose an edge $g=(w,z)\in E(C_{e})$,
and let $h$ be the other edge of $C_{e}$ that is incident to $w$.
Consider a proper partial $\Delta$-edge-coloring of $G$ obtained
from $\theta$ as follows, 
\begin{enumerate}
\item shift the colors on the path $P_{w}$, and clear the color of the
edge that is incident to $z$ and has a color $1$ in $\theta$,
\item color $g$ with color $1$, and color the edges of the even path $C_{e}-g$
with colors $\Delta-1$ and $\Delta$ alternately, starting from the edge
$h$. 
\end{enumerate}

This new partial $\Delta$-edge-coloring
of $G$ corresponds to a maximum $\Delta$-edge-colorable subgraph
$H'$ of $G$, which satisfies $
|E(\overline{F})\cap E(H')|>|E(\overline{F})\cap E(H_g)|=|E(\overline{F})\cap E(H)|$,
contradicting the choice of  $H$. $\square$ \end{proof}

\begin{corollary}\label{2factorExtension}Let $\overline{F}$ be any 2-factor
of a graph $G$ with $\Delta(G)\geq3$. Then, there is a maximum $\Delta(G)$-edge-colorable
subgraph $H$ of $G$, such that $E(\overline{F})\subseteq E(H)$. \end{corollary}

\section{Maximum $\Delta$-edge-colorable subgraphs: Cuts and size}

\begin{theorem}\label{cutCondition} Let $H$ be any maximum $\Delta(G)$-edge-colorable
subgraph of a graph $G$. Then 
\begin{enumerate}
\item $|\partial_{H}(X)|\geq\lceil\frac{|\partial_{G}(X)|}{2}\rceil$ for each $X\subseteq V(G)$,
\item $d_{H}(x)\geq\lceil\frac{d_{G}(x)}{2}\rceil$ for each vertex $x$ of $G$, and
\item $\delta(H)\geq\lceil\frac{\delta(G)}{2}\rceil$.
\end{enumerate}
Furthermore, the bounds are best possible. 
\end{theorem}

\begin{proof} 1. Let $X\subseteq V(G)$, and assume
that $\partial_{G}(X)$ contains $k$ uncolored edges. 
By Theorem \ref{assignment}, there is 
an assignment of colors to uncolored edges with respect to $H$ such
that the corresponding cycles do not intersect. Since every 
cycle $C_{\alpha,\beta}^{e}$ of an uncolored edge $e$ of $\partial_{G}(X)$,
intersects $\partial_{G}(X)$ at least twice, there are at least 
$k$ pairwise different edges of $H$ that belong to $\partial_{G}(X)$.

Statements 2. and 3. follow directly from the first. It remains to show that the bounds are best possible. 
Let $H_r$ be the complete bipartite graph $K_{2r+1,2r+1}$ with a subdivided edge. It is well known, that $H_r$ is 
$(2r+1)$-critical. Therefore it has a $(2r+1)$-edge-coloring,  that leaves precisely one 
edge, which is incident to a vertex of degree 2, uncolored. Take $r$ copies
of $H_r$ and identify the vertices of degree two to obtain the graph $H$. 
Now take two copies of $H$
and connect the vertices of degree $2r$ with an edge to obtain the
graph $G$. It is not hard to observe that $G$ has a maximum $(2r+1)$-edge-colorable subgraph of minimum degree $r+1$. 
This implies that the bounds are best possible (for statement 3 as well, since $G$ is regular). Note, that the graphs are simple.

For multi-graphs, take $k > 1$ copies of the graph with three vertices, one vertex of degree 2 and the other two
of degree $2k$. These graphs have a $2k$-edge-colorable subgraph that leaves precisely one 
edge uncolored, which is incident to a bivalent vertex. Identify the bivalent vertices to obtain a
$2k$-regular graph with the desired properties. 
$\square$
\end{proof}

Next we will construct some graphs to which we will refer in the following to show that the bounds of Theorem \ref{Multigraphgirth} are best possible. Recall that, by definition, the girth of a graph is always at least three.

\begin{lemma} \label{Example}
For each $k \geq 1$, there is a graph $G_k$ with girth $g  = 2k+1$, such that $|E(H_k)| = \frac{2k}{2k+1}|E(G_k)|$, 
for each maximum $\Delta(G_k)$-edge-colorable subgraph $H_k$.
Furthermore, there is a maximum $\Delta(G_k)$-edge-colorable subgraph $H_k^*$ 
with $\Delta(H_k^*) = \lceil \frac{2k}{2k+1} \Delta(G_k) \rceil$.
\end{lemma}

\begin{proof}
For $k \geq 1$, let $G_k = C^{2k}_{2k+1}$ be a cycle of length $2k+1$, where each edge has multiplicity $2k$. 
Then $\Delta(C^{2k}_{2k+1})=4k=\Delta$, and $|E(C^{2k}_{2k+1})|=2k(2k+1)$. Every color class contains at most $k$ edges, hence 
$|E(H_k)| \leq 4k^2$, for every maximum $\Delta$-edge-colorable subgraph. 

Let $H_k^*$ be the subgraph of $G_k$ with
$4k^2$ edges, one edge of multiplicity $2k$ and the remaining $2k$ edges of multiplicity $2k-1$. It is easy to see that 
$H_k^*$ is $4k$-edge-colorable, thus $|E(H_k)| = 4k^2$ for all maximum $\Delta$-edge-colorable subgraphs of $G_k$.
Furthermore $\Delta(H_k^*) = 4k-1  = \lceil 4k(1 - \frac{1}{2k+1}) \rceil$. $\square$
\end{proof}

\begin{theorem}\label{Multigraphgirth} If $G$ is a graph with girth $g \in \{2k, 2k+1\}$ 
($k\geq 1$), and $H$ a maximum $\Delta(G)$-edge-colorable
subgraph of $G$, then $|E(H)| \geq \frac{2k}{2k+1}|E(G)|$, 
and the bound is best possible. 
\end{theorem}

\begin{proof} Let $H$ be a maximum $\Delta(G)$-edge-colorable subgraph of $G$, and 
$\{e_1, \dots, e_n\}$ be the set of uncolored edges. Let $C_{\alpha_i, \beta_i}^{e_i}$ be the pairwise edge-disjoint 
cycles of Theorem \ref{assignment}, $G'=(V(G), \bigcup_{i=1}^n E(C_{\alpha_i, \beta_i}^{e_i}))$, and $H'$ the colored
subgraph of $G'$. Then it follows by Theorem \ref{assignment}, that $|E(H')| \geq \frac{2k}{2k+1}|E(G')|$. Furthermore,
$E(G)- E(G')=E(H)- E(H')$, and hence $|E(H)|= |E(H')|+|E(G)-E(G')| \geq \frac{2k}{2k+1}|E(G')| + |E(G)-E(G')| \geq \frac{2k}{2k+1}|E(G)|$.
The graphs $G_k$ of  Lemma \ref{Example} show that the bound is best possible. 
$\square$
\end{proof}

Theorem \ref{cutCondition} implies that $\Delta(H) \geq \lceil \frac{\Delta(G)}{2} \rceil$. 
However for the maximum vertex degree of a maximum $\Delta(G)$-edge-colorable subgraph $H$ of $G$ much better 
bounds can be proved. For this we need the well known fact, that 
(*) $\chi'(G) \leq  \frac{g}{g-1}\Delta(G) + \frac{g-3}{g-1}$ for graphs $G$ with $\chi'(G) > \Delta(G) +1$ 
and odd girth $g$. This result is easily deducible from the results of Kierstead \cite{Kierstead} (about acceptable paths), 
and it implies that 
$\Delta(G) = \chi'(H) \leq \frac{g}{g-1}\Delta(H) + \frac{g-3}{g-1}$. In \cite{girthEckhard} it is proved 
that $\chi'(G) \leq \Delta(G) + \lceil \frac{\mu(G)}{\lfloor \frac{g}{2}\rfloor}\rceil$. We summarize the consequences of these
two results in the following corollary.  

\begin{corollary} If $G$ is a graph with girth $g \in \{2k, 2k+1\}$ 
$(k\geq 1)$, and $H$ a maximum $\Delta(G)$-edge-colorable subgraph of $G$, then
\begin{enumerate}
\item $\chi'(H) \geq \max \{\chi'(G) - \lceil \frac{\mu(G)}{k}\rceil, \frac{2k}{2k+1}\chi'(G) - \frac{2k-2}{2k+1}\}$,
\item $\Delta(H) \geq \max \{ \Delta(G) - \lceil \frac{\mu(G)}{k}\rceil, \frac{2k}{2k+1}\Delta(G) - \frac{2k-2}{2k+1} \}$.
\end{enumerate}
\end{corollary}

\subsection{Simple graphs}

\begin{theorem}\label{MatchingComplement} Every simple graph
$G$ contains a maximum $\Delta$-edge-colorable subgraph, such
that the uncolored edges form a matching. \end{theorem}

\begin{proof}
Take a maximum $\Delta$-edge-colorable subgraph $H$ of $G$ which
minimizes the number of pairs of adjacent uncolored edges. The proof
will be completed if we show that this number is zero.

Assume, on the contrary, that there is a vertex $v$ which is incident
to two uncolored edges, one of which is $(u,v)$. Let $H$ be $\Delta$-edge-colored
and assume that $\alpha_{0}\in\overline{C}(v)$, $\beta\in\overline{C}(u)$. Since $H$
is a maximum $\Delta$-edge-colorable subgraph of $G$, there is an
edge $(u,v_{0})$ that is colored with color $\alpha_{0}$. Consider a maximal
fan beginning from the vertex $v_{0}$, that is a maximal sequence
$(v_{0},\alpha_{0},v_{1},\alpha_{1},...,v_{k},\alpha_{k})$ such that 

\begin{enumerate}
\item $(u,v_{i})\in E(G)$, and for $i=0,...,k$, the edge $(u,v_{i})$ is colored with color $\alpha_{i}$, 
\item $\alpha_{i}\in\overline{C}(v_{i-1})$, for  $i=1,...,k$, and 
\item $\alpha_0, \ldots ,\alpha_k$ are distinct colors. 
\end{enumerate}

We show that $\overline{C}(v_{k})\neq\emptyset$. Suppose that $\overline{C}(v_{k})=\emptyset$.
This means that $v_{k}$ is a vertex of maximum degree $\Delta(G)$
and all edges incident to it are from $H$. Consider a maximum $\Delta$-edge-colorable
subgraph $H'$ of $G$ obtained from $H$ as follows: Color $(u,v)$
by $\alpha_{0}$, for $i=0,...,k-1$ color $(u,v_{i})$ with color $\alpha_{i+1}$, and leave the edge
$(u,v_{k})$ uncolored. Then there are less pairs of adjacent
uncolored edges in $H'$ than in $H$, which
contradicts the choice of $H$.

Thus, there exist $\alpha_{k+1}\in\overline{C}(v_{k})$. The maximality
of the sequence $(v_{0},\alpha_{0},...,v_{k},\alpha_{k})$ implies that $\alpha_{k+1}\in\{\alpha_{0},...,\alpha_{k-1}\}$,
say $\alpha_{k+1}=\alpha_{i}$, for an $0\leq i\leq k-1$.
Recall that $\beta \in\overline{C}(u)$. The maximality of $H$ implies that
$\beta \in C(v_{k})$. Consider the maximal $\beta$-$\alpha_{k+1}=\beta$-$\alpha_{i}$ alternating
path $P$ beginning from the vertex $v_{k}$. We will show that our assumptions imply
that $H$ is not a maximum $\Delta$-edge-colorable subgraph of $G$, contradicting the choice of $H$.

Case 1: $P$ does not reach $u,v_{i},v_{i-1}$.
The following recoloring yields the desired contradiction.
Color the edge $(u,v)$ by $\alpha_{0}$, $(u,v_{j})$ by $\alpha_{j+1}$, $j=0,...,k-1$,
exchange the colors on the path $P$ and color $(u,v_{k})$ by $\beta$.

Case 2: $P$ reaches $v_{i}$.
Since the edge $(u,v_{i})$ is colored by $\alpha_{i}$ it follows 
that the path $P$ enters $v_{i}$ by an edge colored $\beta$. To get the contradiction, color
the edge $(u,v)$ by $\alpha_{0}$, $(u,v_{j})$ by $\alpha_{j+1}$, $j=0,...,i-1$,
exchange the colors on the path $P$ and color $(u,v_{i})$ by $\beta$.

Case 3: $P$ reaches $v_{i-1}$.
Since $\alpha_{i}\in\overline{C}(v_{i-1})$, it follows that the path
$P$ reaches $v_{i-1}$ by an edge colored $\beta$. Color the edge
$(u,v)$ by $\alpha_{0}$, $(u,v_{j})$ by $\alpha_{j+1}$, $j=0,...,i-2$,
exchange the colors on the path $P$ and color $(u,v_{i-1})$ by $\beta$ to get the desired
contradiction. $\square$ \end{proof}

Theorem \ref{MatchingComplement} is equivalent to 

\begin{theorem} If $G$ is a simple graph, then $r_{e}(G)=r'_{e}(G)$.
\end{theorem}

Lemma \ref{Example} shows that a maximum $\Delta$-edge-colorable subgraph of a multigraph can be class II as well. 
This cannot be the case for simple graphs as the following theorem shows.

\begin{theorem} \label{equality} If $H$ is a maximum $\Delta(G)$-edge-colorable
subgraph of a simple graph $G$, then $\Delta(H)=\Delta(G)$, i.e. $H$ is class I.
\end{theorem}

\begin{proof} Let $e = (v,w) \in E(G)-E(H)$ be an uncolored edge. Then $\chi'(H+e) = \Delta(G)+1$, and hence, since $H+e$ is simple, $\Delta(H+e)= \Delta(G)$. The graph $H+e$ contains a 
$\Delta(G)$-critical subgraph $H'$, which clearly contains the edge $e$. By Vizing's Adjacency Lemma \cite{Vizing}, every vertex 
of $H'$ is adjacent to at least two vertices of maximum degree $\Delta(G)$. 
Thus there is a vertex $x \not= v,w$ of $H'$ with $d_{H'}(x)=\Delta(G)$, and hence 
$\Delta(H) = \Delta(G)$. $\square$

\end{proof}

\section{Discussion and some Conjectures}

Theorem \ref{MatchingComplement} says, that every simple class II graph $G$ has a 
maximum $\Delta$-edge-colorable subgraph $H$, such that $\chi'(G \backslash E(H))=1$. 
We believe that this can be generalized to multigraphs.

\begin{conjecture} \label{k_matching} If $G$ is a graph with $\chi'(G)=\Delta(G)+k$ 
$(k \geq 0)$,
then there is a maximum $\Delta(G)$-edge-colorable subgraph $H$ of $G$,
such that $\chi'(G\backslash E(H))=k$. 
\end{conjecture} 

This conjecture is equivalent to the following statement.

\begin{conjecture} Let $G$ be a graph, then $r_{e}(G)=r'_{e}(G)$. \end{conjecture}

\begin{acknowledgement}
We would like to thank our reviewers for their useful comments that helped us to improve the presentation of the paper. 
Our special thanks to the second reviewer, who has pointed out a mistake in the earlier version of the proof of Lemma \ref{CyclesIntersection}.
\end{acknowledgement}

\end{document}